\documentclass[conference]{IEEEtran}

\makeatletter
\def\ps@headings{%
\def\@oddhead{\mbox{}\scriptsize\rightmark \hfil \thepage}%
\def\@evenhead{\scriptsize\thepage \hfil \leftmark\mbox{}}%
\def\@oddfoot{}%
\def\@evenfoot{}}
\makeatother

\usepackage[square, numbers, comma, sort&compress]{natbib}  
\usepackage{amsmath,amsfonts,amssymb,amscd,amsthm,xspace}
\usepackage{color}
\usepackage{graphicx}
\usepackage{epstopdf}
\usepackage{amsmath,amssymb,latexsym}

\usepackage[scriptsize]{subfigure}

\usepackage{setspace}

\usepackage{algorithmic}
\usepackage{multirow}

\usepackage{subfigure}
\usepackage[T1]{fontenc}
\usepackage[utf8]{inputenc}
\usepackage{authblk}

\newtheorem{definition}{Definition}
\newtheorem{theorem}{Theorem}

\newtheorem{lemma}{Lemma}

\theoremstyle{remark}
\newtheorem{remark}{Remark}

\begin{document}
\newcommand{\sj}[1]{\textcolor{black}{\textbf{SJ:} #1}}
\newcommand{\howard}[1]{\textcolor{black}{#1}}

\newcommand{\norm}[1]{\left\| #1 \right\|_1}
\newcommand{\size}[1]{\left| #1 \right|}
\newcommand{\abs}[1]{\left| #1 \right|}
\newcommand{\entropy}[1]{H\left( #1 \right)}
\newcommand{\condentropy}[2]{H\left( {#1} | {#2} \right)}
\newcommand{\mutual}[2]{I\left( {#1} ; {#2} \right)}
\newcommand{\condmutual}[3]{I\left( {#1} ; {#2} | {#3} \right)}

\newcommand{\graph}{\mathcal{G}}
\newcommand{\vertex}{\mathcal{V}}
\newcommand{\edge}{\mathcal{E}}
\newcommand{\src}{s}
\newcommand{\term}{t}
\newcommand{\inode}{\overline{\mathcal{V}}}
\newcommand{\Pset}{\mathcal{P}}
\newcommand{\Mset}{\mathcal{M}}
\newcommand{\rate}{R}
\newcommand{\const}{k}
\newcommand{\blk}{n}
\newcommand{\oblk}{N}
\newcommand{\Enc}{Enc}
\newcommand{\Dec}{Dec}
\newcommand{\In}[1]{\edge_{in}(#1)}
\newcommand{\Out}[1]{\edge_{out}(#1)}

\newcommand{\Cp}{C}
\newcommand{\Cpn}{C_{node}}
\newcommand{\Fq}{\mathbb{F}_q}

\newcommand{\dg}{\Gamma}
\newcommand{\Flow}{F}

\newcommand{\Oh}{\mathcal{O}}

\newcommand{\Zn}{\mathcal{Z}}
\newcommand{\Msg}{\mathbf{M}}
\newcommand{\Key}{\mathbf{K}}
\newcommand{\msg}{m}
\newcommand{\key}{k}
\newcommand{\code}{x}
\newcommand{\vanmtrx}{\mathbf{V}}
\newcommand{\rk}{r}

\newcommand{\vm}{\mathbf{m}}
\newcommand{\vk}{\mathbf{k}}
\newcommand{\x}{\mathbf{x}}
\newcommand{\X}{\mathbf{X}}
\newcommand{\Y}{\mathbf{Y}}
\newcommand{\Z}{\mathbf{Z}}
\newcommand{\Pvar}{\mathbf{P}}
\newcommand{\T}{\mathbf{T}}
\newcommand{\U}{\mathbf{U}}
\newcommand{\D}{\mathbf{D}}

\newcommand{\Code}{\mathcal{C}}
\newcommand{\Act}{\mathcal{A}}

\newcommand{\linktapper}{z_{link}}
\newcommand{\nodetapper}{z_{node}}

\newcommand{\ncut}{\mathcal{N}}

\title{Routing for Security in Networks with Adversarial Nodes}


\author{{Pak Hou Che}$^{\star}$, {Minghua Chen}$^{\star}$, {Tracey Ho}$^{\dagger}$,
{Sidharth Jaggi}$^{\star}$, and {Michael Langberg}$^{\ddagger}$ \vspace{2mm} \\
$^{\star}$Department of Information Engineering, The Chinese University
of Hong Kong\\
$^{\dagger}$Department of Electrical Engineering,
California Institute of Technology \\
$^{\ddagger}$Department of Mathematics and Computer Science, The Open University of Israel}

\maketitle

\begin{abstract}
\footnote{The authors are listed in alphabetical order.}
We consider the problem of secure unicast transmission between two nodes in a directed graph, where an adversary eavesdrops/jams a subset of nodes. This adversarial setting is in contrast to traditional ones where the adversary controls a subset of links. In particular, we study, in the main, the class of routing-only schemes (as opposed to those allowing coding inside the network). Routing-only schemes usually have low implementation complexity, yet a characterization of the rates achievable by such schemes was open prior to this work. We first propose an LP based solution for secure communication against eavesdropping, and show that it is information-theoretically rate-optimal among all routing-only schemes. The idea behind our design is to balance information flow in the network so that no subset of nodes observe ``too much'' information. Interestingly, we show that the rates achieved by our routing-only scheme are always at least as good as, and sometimes better, than those achieved by ``na\"ive'' network coding schemes ({\it i.e.} the rate-optimal scheme designed for the traditional scenario where the adversary controls links in a network rather than nodes.) We also demonstrate non-trivial network coding schemes that achieve rates at least as high as (and again sometimes better than) those achieved by our routing schemes, but leave open the question of characterizing the optimal rate-region of the problem under all possible coding schemes. We then extend these routing-only schemes to the adversarial node-jamming scenarios and show similar results. During the journey of our investigation, we also develop a new technique that has the potential to derive non-trivial bounds for general secure-communication schemes.
\end{abstract}
\IEEEpeerreviewmaketitle

\section{Introduction}

The secure network coding problem, introduced by Cai and Yeung \cite{cai02secure}, considers communication of a secret message in the presence of a computationally-unlimited adversary that eavesdrops on a limited but unknown portion of the network.  Most existing work in the literature concerns the multicast uniform link-based adversary case, where all links have equal capacity and the adversary can eavesdrop on a limited number of links. In this case, the maximum secure rate achievable when only the source generates randomness has a simple cut-set characterization~\cite{cai02secure}, and is achieved by a number of existing coding schemes, {\it e.g.}~\cite{feldman04capacity,silva11universal,elrouayheb12secure}.

In this paper we consider the node-based adversary case, where a computationally-unlimited adversary can eavesdrop on a limited number of nodes.  Much less is known about this problem. Motivated by complexity considerations, we focus on the class of routing-only schemes for unicast, in which only the source performs coding while non-source nodes perform routing. We formulate a linear program (LP) that balances the amount of information flowing through any subset of nodes, and show that its solution, which  involves only simple forwarding, achieves the optimal capacity within the class of routing-only schemes.  This class includes schemes involving replication (transmitting multiple copies of a received packet); our result shows that such replication does not improve rate. We further show that our LP-based routing-only schemes achieve rates that are always at least and sometimes higher than rates achieved by na\"ive application of secure network coding schemes designed for the uniform link-adversary case.
Related work by Cui {\it et al.}~\cite{cui10-1} considers the link-based secrecy problem with unequal link capacities and/or restricted eavesdropping sets, and give some achievable coding schemes where random keys may be injected or canceled at intermediate nodes. We apply these approaches to the node-based eavesdropping problem and show that they can sometimes achieve higher rates than our routing-only schemes, though at the expense of higher complexity.

We further extend our routing-only schemes to the problem of coding against a node-based jamming adversary that can introduce arbitrary errors at nodes under his control. The problem of network error correction coding against a jamming adversary was introduced by Yeung and Cai~\cite{yeung06network,cai06network}.   Like the eavesdropping problem, network error correction for the multicast uniform link-based adversary case has been extensively studied, with various existing capacity-achieving code constructions {\it e.g.}~\cite{cai06network,jaggi_resilient_08,RankMetricRandCodes}, while much less is known about the node-based adversary case. Similarly, we show that our routing-only schemes, obtained using the same LP formulation, achieve rates that are never lower and sometimes higher compared to that achieved by na\"ive application of network error correction codes designed for the uniform link-adversary case. However, unlike the eavesdropping case,  we show that replication can improve rate in the jamming case. Kosut {\it et al.}~\cite{kosut09} also consider node-based jamming adversaries, and introduce non-linear network codes called ``polytope codes" in which intermediate nodes carry out comparison and signaling operations. These codes can sometimes achieve higher rates than routing-only schemes, but are more complex.

{
One ``natural'' restriction we consider in the jamming scenario, in contrast to most work in the network error-correction literature, is that the adversary is ``causal''. That is, his jamming actions cannot be based on future transmissions on the network. Under this reasonable assumption, we note that the power of the adversary is significantly weakened compared to the ``non-causal'' scenario. Specifically, we show that ideas in~\cite{jaggilangberg05isit} lead to code designs in which the same rates can be achieved against a {\it causal omniscient} adversary (one who can see all causal transmissions in the network, and base his jamming strategy as a function of these observations), as are achieved by our schemes against a {\it localized} adversary (one who can see only see transmissions on edges incoming to him, and base his jamming strategy as a function of these observations).}

%
%

\subsection{Notational Conventions}

Calligraphic symbols such as $\mathcal{N}$ will denote sets. Boldface symbols such as $\mathbf{x}$ will denote vectors, boldface upper-case symbols such as $\mathbf{X}$ will denote random variables, non-boldface lower-case symbols such as $x$ will denote particular instantiations of those random variables and non-boldface upper-case symbols such as $X$ will denote matrices.

\section{Model}

\subsection{Network Model}

Let a graph $\graph = (\vertex, \edge)$, where $\vertex$ is the vertex set, and $\edge$ is the edge set. There are two pre-specified nodes in $\vertex$ -- specifically $\src$ denotes the {\it source node}, and $\term$ denotes the {\it terminal node}. For notational convenience, we denote by $\inode$ the set of {\it internal nodes} $\vertex \setminus \{ \src,\term \}$, {\it i.e.}, the subset of nodes of $\vertex$ excluding the source and terminal nodes. As is common in the network coding literature~\citep{KoeM_03}, we assume each edge has unit capacity.\footnote{In the node-adversary case this unit-capacity assumption is without loss of generality (not so in the case when the adversary controls edges -- {see, for instance,~\cite{cui10-1}}).} For any nodes $v \in \inode$, let $\In{v}$ denote the {\it set of incoming edges} of node $v$ and $\Out{v}$ denote the {\it set of outgoing edges} of node $v$. We also define $\In{\mathcal{A}}$ and $\Out{\mathcal{A}}$ be the set of incoming and outgoing edges of the nodes $v \in \mathcal{A}$ respectively. For directed edge $e =(v,v') \in \edge$, let $head(e)$ denote the head node of the edge $e$, {\it i.e.}, $head(e)=v'$, and $tail(e)$ denote the tail node of the edge $e$, {\it i.e.}, $tail(e)=v$. The {\it min-cut of the network between the source $\src$ and the terminal $\term$} is denoted by $\Cp$.


\subsection{Source Encoding}

A {\it packet} is defined as a length-$\blk$ vector in the field $\Fq$. Here the {\it field-size} $q$, the {\it number of packets in a generation} $\oblk$, the {\it rate} $\rate$, the {\it redundancy} $\delta$, and the {\it key rate} $\rk$ are code-design parameters to be specified later. We also define $\tau$ to be the {\it generation length}, which satisfies $\oblk \leq \tau \Cp$, {\it i.e.}, the number of packets in a generation is at most the generation length times the min-cut. A visual presentation of these parameters are given in Figure~\ref{fig:illustration}. The source $\src$ has a {\it message} $\Msg$ drawn arbitrarily from the set $\{1,2,\ldots,q^{\rate \oblk \blk (1 - \delta)}\}$, and a random variable {\it key} $\Key$ distributed uniformly from the set $\{1,2,\ldots,q^{\rk \oblk \blk (1 - \delta)}\}$. The source $\src$ then encodes the message $\Msg$ and the key $\Key$ by the {\it source encoder} $\Enc(\src)$, and generates $\oblk \blk$ symbols over $\Fq$, {\it i.e.}, $\Enc: \{1,2,\ldots,q^{\rate \oblk \blk (1 - \delta)}\} \times \{1,2,\ldots,q^{\rk \oblk \blk (1 - \delta)}\} \rightarrow \{1,2,\ldots,q^{(\rate + \rk) \oblk \blk}\}$.

\begin{figure}
  \centering
  \includegraphics[width=1\columnwidth]{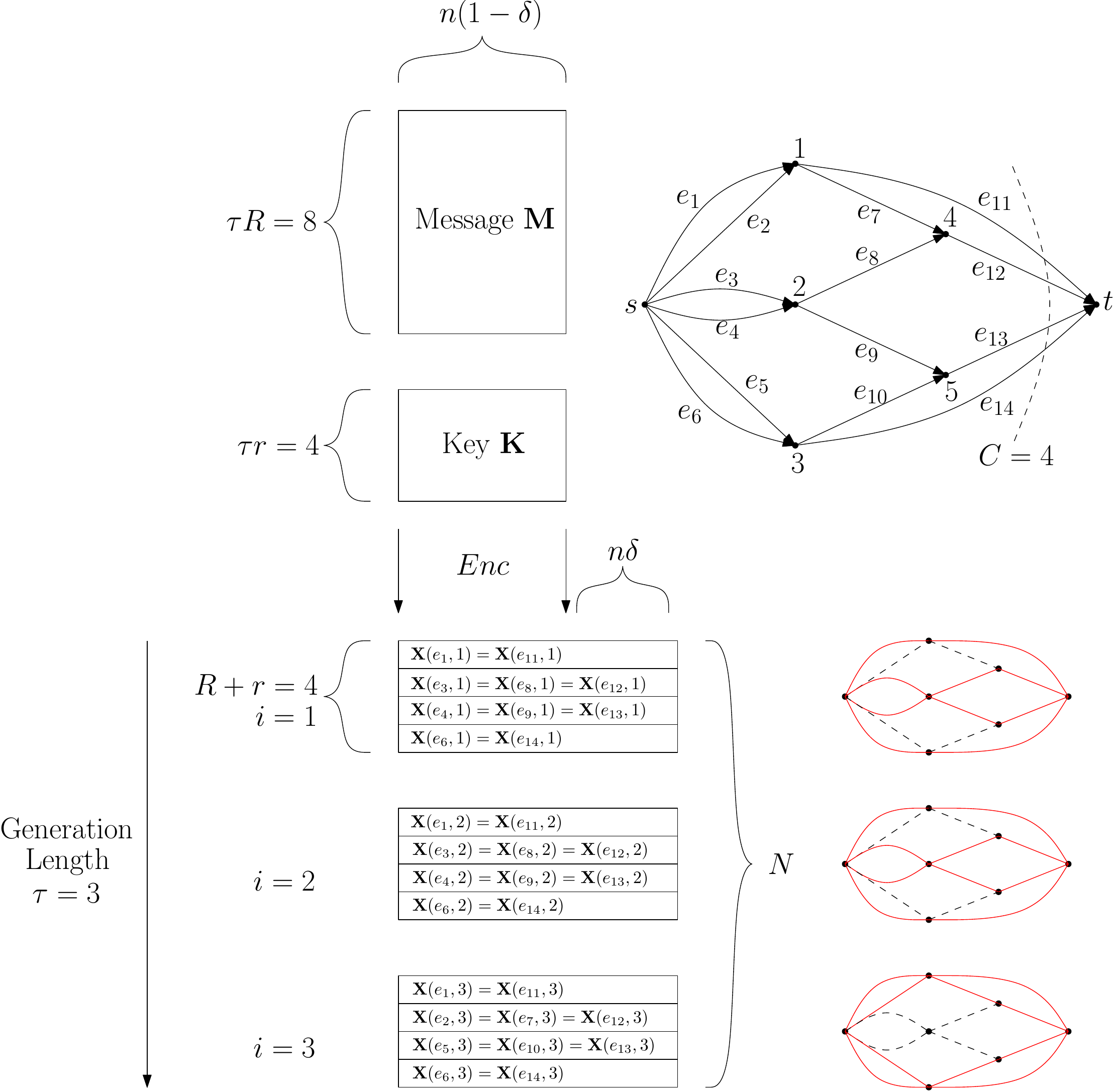}
  \caption{{\bf Illustrating example for our code parameters:} The source $\src$ wishes to transmit a messge $\Msg$ to the terminal $\term$ over a network $\graph = (\edge,\vertex)$ with min-cut $\Cp$ (in this example $\Cp = 4$), specifically the so-called ``cockroach network'' example first described in~\cite{kosut09}, and replicated on the upper right of this figure. To this end, it first organizes $\Msg$ into $\tau\rate = 8$ packets (in this example, the generation length $\tau = 3$, and the rate $\rate=8/3$), each containing $\blk(1-\delta)$ symbols over $\Fq$. It then generates a uniformly random key $\Key$ which it organizes into $\tau\rk$ packets (in this example $\rk = 4/3$), each containing $\blk(1-\delta)$ symbols over $\Fq$. Next, the source uses $\Enc$ to encode $\Msg$ and $\Key$ into $\oblk$ packets (in this example $\oblk = 12$), each containing $\blk$ symbols over $\Fq$. In each coding instant $i$ within the generation of length $\tau$ the source then injects at most $\Cp$ of these packets into the network (in this example $i \in \{1,2,3\}$, the outputs of the encoder are denoted $\X(e,i)$, for appropriate $e$ and $i$, and routed over the network according to the red paths denoted in the three figures on the right). Finally, the terminal uses $\Dec$ to decode $\Msg$ as $\hat{\Msg}$. The set of all node encoders, along with the decoder, together comprise the code $\Code$.}
  \label{fig:illustration}
\end{figure}

\subsection{Linear Network Encoding}

\footnote{In some models, non-linear coding outperforms linear coding \citep{kosut09}. For complexity reasons, we restrict our attention to linear codes.}There are three types of nodes in the network -- ``uncorrupted nodes'', ``eavesdropping nodes'', and ``jamming nodes". Nodes in the first category are entirely honest, perform the encoding operations specified in this section, and do not aim to eavesdrop on communications. Nodes in the second category also perform the encoding operations specified in this section, but in addition attempt to eavesdrop on communication as specified in Section~\ref{sec:model} 1a. Nodes in the third category do not perform the encoding operations specified in this Section (their ``jamming'' is described in Section~\ref{sec:model} 1b and 2a), and in addition also attempt to eavesdrop on communications. We shall call nodes in either of the first two categories ``non-jamming''.

The random variable $\X(e,i)$ denotes the packet on edge $e \in \edge$ at time $i \in \{1, \ldots ,\tau \}$. {For simplicity, we sometimes omit the time index, and use $\X(e)$ to denote the set of {\it all} packets going over an edge in a generation.} We also denote $\X(\edge',i)$ to be set of packets
$\{ e \in \edge': \X(e,i) \}$ at time $i \in \{1, \ldots ,\tau \}$, where $\edge' \subseteq \edge$.

Each non-jamming node in the network also has an encoder. As mentioned before, in this work we restrict the internal nodes in the network to ``simple" operations, specifically causal linear operations\footnote{In most of the network coding literature, we do not explicitly worry about causality, since a ``limited'' amount of non-causality can be simulated by pipelining (buffering at each node). {\bf However, in adversarial jamming problems the throughout against a causal adversary can be higher than against a noncausal adversary. In this work, this is indeed the case in the Omniscient Jammer model.
Hence we explicitly focus on causal adversaries.}} over $\Fq$. That is, the packets transmitted on each outgoing edge of a node $v$ are linear functions of the packets arriving on incoming edges of $v$.



We distinguish two types of network encoding schemes:

\noindent {\bf Routing schemes:} In a {\it routing scheme}, the set of packets leaving a node $v$ are subsets of packets incoming to that node. That is, any packet $\X(e,i)$ transmitted on an edge $e \in \Out{v}$ at time $i \in \{1,\ldots,\tau\}$ equals a packet $\X(e',j)$ transmitted on an edge $e' \in \In{v}$ at time $j \leq i$. Note that this includes ``replication'', {\it i.e.}, a node is allowed to transmit multiple copies of a packet it has observed.


\noindent {\bf Coding schemes:} In a {\it coding scheme}, the set of packets leaving a node $v$ are linear combinations of packets incoming to that node\footnote{In this model we disallow the possibility that an internal node in the network generates private randomness, and uses this to generate outgoing packets. It can be shown {(see~\cite{cui10-1}, and Figure~\ref{fig:int_rand} in Section~\ref{sec:beyroute})} that in fact such a strategy can sometimes increase the throughput of networks.}. These linear combinations can be of two types. In {\it scalar linear network coding schemes}, each outgoing packet corresponds to a causal linear combination (over $\Fq$) of the packets that $v$ has already observed. That is, for any packet $\X(e,i)$ with $tail(e) \in \inode$, we have
\begin{equation}
	\X(e,i) = \sum_{j \leq i} \sum_{e':head(e')=tail(e)} \beta(e',e,j) \X(e',j),
\end{equation}
\noindent where the linear network {\it coding coefficients} $\beta(e,e',j)$ are scalars from $\Fq$.

In {\it vector linear network coding schemes}, each symbol of each outgoing packet corresponds to a linear combination (over $\Fq$) of all the symbols of all the packets that $v$ has already observed. That is, for any packet $\X(e,i)$ with $tail(e) \in \inode$, we have
\begin{equation}
	\X(e,i) = \sum_{j \leq i} \sum_{e':head(e')=tail(e)} B(e',e,j) \X(e',i)
\end{equation}
\noindent where $B(e',e,j)$ are matrices in $\Fq^\blk \times \Fq^\blk$. In particular, if $B(e',e,j) = \beta(e',e,j) I$, it is a scalar linear network coding scheme.\footnote{Vector linear network coding schemes are more general than scalar linear network coding schemes -- see \cite{Jaggi:04Allerton}.
In general, all the achievability schemes we present in this paper are based on scalar linear network coding schemes. However, some of the non-achievability results we present work even for vector linear network coding schemes.}

For both these types of codes, the choice of coding coefficients is part of the code design, and is explicitly specified later in the various schemes we construct. In general they may be chosen either deterministically (as a function of $\graph$) or randomly\footnote{Each node chooses its linear network coding coefficients uniformly at random over $\Fq$, for instance~\cite{random_nc06}.}.  We define the {\it network code} $\Code$ to be a triple that contains source encoder $\Enc(\src)$, intermediate node encoders $\Enc(v)$ for all $v \in \inode$ and terminal decoder $\Dec(\term)$. That is, $\Code = (\Enc(\src), \Enc(\inode), \Dec(\term))$ -- here $\Enc(\inode)$ is $\Enc(v)$ where $v \in \inode$.

\subsection{Adversarial Models and Corresponding Communication Goals}
\label{sec:model}
We focus on two broad classes of adversarial models -- localized and omniscient adversaries, and their corresponding communication goals. Localized adversaries are usually considered as the adversaries in the wired model, omniscient adversaries are usually considered as the adversaries in the wireless model.

\noindent 1) {\bf Localized Adversaries:} An adversary is said be to {\it localized} if it only has a casual ``localized'' view of network traffic, depending on the nodes in $\Zn$ it controls. That is, a {\it localized adversary that observes $\Zn$} can observe the packets incoming to the set of nodes $\Zn$. Its ``attack strategy" can be a causal function of these observations (and also its knowledge of $\graph$ and $\Code$, and the terminal's decoding function, {as} defined below).

We consider three types of communication problems against localized adversaries:

\noindent a) {\bf Eavesdropping:} The {\it set of nodes eavesdropped by the adversary} $\Zn_E$ is a set of at most $z_E$ nodes in $\inode$, chosen by the adversary as a function of his knowledge of $\graph$ and $\Code$, prior to communication starting. That is, $\Zn_E \subseteq \inode: \graph \times \Code \rightarrow \Pset_{z_E}(\inode)$, where $\Pset_{z_E}(\inode)$ denotes the set of all subsets of $\inode$ of size less than or equal to $z_E$. Given this choice, at time $i$ the adversary observes packets $\X(\In{\Zn_E},j)$ with $j \leq i$, the information on edges incoming to nodes in $\Zn_E$ at time $j \leq i$. Given these packets, the adversary's estimate $\hat{\Msg}$ of $\Msg$ is allowed to be an arbitrary (possibly probabilistic) function of the packets he observes, the network $\graph$, and the network code $\Code$.
\noindent {\it Adversarial Communication Goals Against a Localized Eavesdropper:} Prior to the communication commencing, both $\Msg$ and $\Key$ are known only to the source $\src$ itself, and not to any other party. $\src$ wishes to transmit the message $\Msg$ to $\term$ over the network $\graph$, such that the secrecy and decodability requirements described in (\ref{eq:sec}) and (\ref{eq:dec}) in \ref{sec:dec} below are satisfied.

\noindent b) {\bf Jamming:} The {\it set of nodes jammed by the adversary} $\Zn_J$ is a set of at most $z_J$ nodes in $\inode$. Given this choice, at time $i$ the adversary can access $\X(\In{\Zn_J},j)$ with $j \leq i$. {Given the network $\graph$ and the network code $\Code$, he then corrupts the information of the outgoing links of $\Zn_J$, that he replaces $\X(\Out{\Zn_J},i)$ by $\widehat{\X}(\Out{\Zn_J},i)$ for all $i \in \{1,...,\tau\}$. The adversary's transmissions $\widehat{\X}(e,i)$ on edges $e$ outgoing from nodes in $\Zn_J$ are allowed to be arbitrary (possibly probabilistic) {\it casual} functions of the packets he observes, the network $\graph$, and the network code $\Code$.}
	
\noindent {\it Adversarial Communication Goals Against a Localized Jammer:} In this scenario, $\src$ wishes to transmit the message $\Msg$ to $\term$ over the network $\graph$, such that the decodability requirement described in (\ref{eq:dec}) is satisfied.


\noindent c) {\bf Eavesdropping and Jamming:} The set of nodes eavesdropped and jammed by the adversary $\Zn$ is a set of at most $z$ nodes in $\inode$. Given the network $\graph$ and the network code $\Code$, he corrupts the information of the outgoing links of $\Zn$ which is the same as the {\bf Localized Jamming} case. Furthermore, the source $\src$ also wishes the message is secure to the adversarial nodes $\Zn$ which has the same setting as the {\bf Localized Eavesdropping} case.

\noindent {\it Adversarial Communication Goals Against a Localized Eavesdropper/Jammer:} $\src$ wishes to transmit the message $\Msg$ to $\term$ over the network $\graph$, such that the secrecy and decodability requirements described in \howard{\eqref{eq:dec} and \eqref{eq:sec}} in \ref{sec:dec} are satisfied.

\noindent 2) {\bf Causal Omniscient Adversaries:} An adversary is said to be {\it causal omniscient} if it has a ``global but causal" view of the network traffic. That is, a {\it causal omniscient adversary} that observes all the information $\X(e,i)$ transmitted over every edge $e$ and all time $i$, though its jamming can only be a causal function in $i$.\footnote{In fact, a secrecy constraint does not make sense in the case of omniscient adversaries, since adversaries by definition know all transmissions in the entire network.} Its ``attack strategy" can be a causal function of these observations (and also its knowledge of $\graph$ and $\Code$).

\noindent a) {\bf Jamming:} Given the information transmitting over the network $\graph$, at time $i$ the adversary can access $\X(e,j)$ with $e \in \edge$ and $j \leq i$. The {\it set of nodes jammed by the adversary} $\Zn$ is a set of at most $z_J$ nodes in $\inode$. Given this and the network $\graph$, the network code $\Code$, he then corrupts the information of the outgoing links of $\Zn$, that is, replace $\X(\Out{\Zn},i)$ by $\widehat{\X}(\Out{\Zn},i)$.
		
\noindent {\it Adversarial Communication Goals Against a Omniscient Jammer:} In this case, $\src$ wishes to transmit the message $\Msg$ to $\term$ over the network $\graph$, such that the decodability requirement described in (\ref{eq:dec}) is satisfied.

\subsection{Terminal Decoding}
\label{sec:dec}
{
In each of the four adversarial models above, the communication goals always include the ``decodability" condition. Only the {\bf Localized Eavesdropping} and {\bf Localized Eavesdropping and Jamming} models also include the ``secrecy" condition. The former is defined in \ref{dec:decodability}, and the latter is defined in \ref{dec:secrecy} below.

\begin{enumerate}
	\item\label{dec:decodability} {\bf Decodability:} We define the {\it decoding function} of terminal $\term$ to be $\Dec$, where $\Dec : \{1,2, \ldots, q^{ (\rate+\rk) \oblk \blk } \} \rightarrow \{ 1,2, \ldots, q^{\rate \oblk \blk (1 - \delta)} \}$. Let $\widehat{\Msg} = \Dec(\Enc(\Msg))$ be the message that the terminal $\term$ decodes. The terminal $\term$ \howard{is} required to be able to decode the original message $\Msg$ with arbitrarily high probability. That is, we need

\begin{equation}
	\Pr_{\Act, \Code}(\widehat{\Msg} \neq \Msg) < \epsilon_1.
	\label{eq:dec}
\end{equation}
\noindent for arbitrarily small $\epsilon_1$.
	
	\item\label{dec:secrecy} {\bf Secrecy:} The source $\src$ transmits the message $\Msg$ with {\it $\Delta$-securely} to the terminal $\term$. {That is, we require the mutual information between the source's message and the adversary's estimate of it to be ``small'', that is,}

\begin{equation}
	\mutual{\Msg}{\X({\In{\Zn_E}})} \leq \Delta.\footnote{Intuitively, this inequality means that the communication scheme leaks at most $\Delta$ units of information.}
	\label{eq:sec}
\end{equation}
\noindent In particular, if $\Delta = 0$, we say the message $\Msg$ is {\it perfectly secure}.

\end{enumerate}
}

The {\it overall probability of error}\footnote{{These definitions are for {\it maximal} probability of error (over all messages $\Msg$) and hence also work {\it averaged} over $\Msg$. The converses we prove {\it also} work {\it averaged} over $\Msg$, and hence are also true for the {\it worst-case} $\Msg$.}} $\Pr_e$ of a transmission scheme can be separated into two parts. The {\it probability of decoding error} and the {\it probability of leakage}. {The probability of decoding error, denoted by $\epsilon_1$,  is $\Pr_{\Act, \Code}(\widehat{\Msg} \neq \Msg)$. The probability of leakage error, denoted by $\epsilon_2$, is defined as $\Pr_{\Act, \Code}(\mutual{\Msg}{\X(\In{\Zn})} > \Delta)$.}

\subsection{Code Parameters}
The rate $\rate = \frac{1}{\blk \oblk} \log_q \size{\Mset}$ is {\it achievable} if for any $\epsilon > 0$, there exists $\delta > 0$ such that there is a coding scheme with rate at least $\rate - \delta$ with the overall probability of error $P_e = \Pr_{\Act, \Code}(\widehat{\Msg} \neq \Msg) + \Pr_{\Act, \Code}(\mutual{\Msg}{\X(\In{\Zn})} > \Delta) < \epsilon$ for large enough $\blk \oblk$ and $q$.

\section{Preliminaries}

\subsection{Routing Linear Program}
We first introduce the linear program that gives us a baseline routing scheme for each of the four models above.

{
Let $\Pset$ be the set of all paths from $\src$ to $\term$. For path $p \in \Pset$, a natural internal variable in the {\bf Linear Program 1} (defined in Equations~\eqref{lp1:max} -- \eqref{lp1:balance}) is the {\it flow through path $p$}, denoted by $F(p)$.

\howard{
\noindent {\bf Linear Program 1}
\begin{alignat}{2}
F(z)  = \mbox{ } &	\text{max }	\quad  \sum_{p \in \Pset} \Flow(p) - \lambda(z) \label{lp1:max},\\
	 &  \text{subject to }	\quad  \forall e \in \edge, & & \sum_{p:p \ni e} \Flow(p) \leq 1, \label{lp1:link}\\
					 	\quad & \forall \Zn \subset \inode, | \Zn | \leq z, & & \sum_{p:|p \cap \Zn|>0} \Flow(p) \leq \lambda(z). \label{lp1:balance}
\end{alignat}}

In LP1, the maximum value of the objective function in \eqref{lp1:max} is denoted by $F(z)$. Equation~\eqref{lp1:link} says that the flows passing through a link are bounded by its capacity (which equals $1$). Equation~\eqref{lp1:balance} bounds the flow through any set of nodes with $| \Zn | \leq z$. This flow is bounded from above by $\lambda(z)$ -- the LP attempts to ensure that not {\it too much} flow passes through any set of $z$ nodes, while simultaneously maximizing the overall flow. Here, $\lambda(z)$ is also a variable of LP1. The choice of rate $\rate$ and key-rate $\rk$ for each of our routing scheme depends critically on $\lambda(z)$.

\begin{lemma}
\label{LP1}
If the optimal solution for LP1 with $\sum_{p \in \Pset} \Flow(p) < \Cp$. Then, there is another optimal solution satisfies $\sum_{p \in \Pset} \Flow(p) = \Cp$.
\end{lemma}
\begin{proof}
Suppose the optimal solution of LP1 is $((\forall p \in \Pset, \Flow_0(p)), \lambda_0)$ such that the sum of all flows $\sum_{p \in \Pset} \Flow_0(p) < \Cp$ and let $\Flow_0 = \sum_{p \in \Pset} \Flow_0(p)$. So, the optimal objective function is $\Flow_0 - \lambda_0$. Note that in this network, we can still inject $\Flow_{in} = \Cp - \Flow_0$ fraction of flows into the network since the sum of all flows $\Flow_0 < \Cp$. Then, we have the sum of all flows $\sum_{p \in \Pset} \Flow'(p) = \Cp$. Denote the increment of $\lambda_0$ after the injection of $\Flow_{in}$ to be $\lambda_{in}$, we have $\lambda_{in} \leq \Flow_{in}$. Also, we have $\forall \Zn \subset \inode, \mbox{ } \sum_{p:p \ni v_i, i \in \Zn} \Flow'(p) \leq \lambda_0 + \lambda_{in}$, where $\lambda_{in} \leq F_{in}$. This means, the increment of the flows that passing through $\Zn$ is $\lambda_{in}$. So, the objective function after the flow injection is $\Cp - (\lambda_0 + \lambda_{in}) \geq \Cp - (\lambda_0 + F_{in}) = F - \lambda_0$. Since $F - \lambda_0$ is optimal, and so as $\Cp - (\lambda_0 + \lambda_{in})$.
\end{proof}
By Lemma~\ref{LP1}, LP1 can be reduced into the following linear program.
\noindent {\bf Linear Program 1'}
\begin{alignat*}{2}
	\text{max }	\quad & \Cp - \lambda(z) \\
	\text{subject to }	\quad & \forall e \in \edge, & & \sum_{p:p \ni e} \Flow(p) \leq 1 \\
						\quad & \forall \Zn \subset \inode, & & \sum_{p:|p \cap \Zn|>0} \Flow(p) \leq \lambda(z) \\
						\quad & \sum_{p \in \Pset} \Flow(p) = \Cp
\end{alignat*}
Note that the size of $\Pset$ is exponential to the network size, that means, there are exponential number of variables. In order to reduce the complexity of solving the linear program, we then consider the following linear program which is equivalent to LP1'. So, we use the standard form of linear program as max-flow min-cut theorem. That is, instead of using the flow on the paths $\Flow(p)$ where $p \in \Pset$ as the variables, we use the flow on the edges $\Flow(e)$ where $e \in \edge$ to be the variables in the following linear program.

\noindent {\bf Linear Program 2}
\begin{alignat*}{2}
	\text{max }	\quad & \Cp - \lambda(z) \\
	\text{subject to }	\quad & \forall v \in \inode, \sum_{e: e \in \In{v}}\Flow(e) = \sum_{e: e \in \Out{v}} \Flow(e) \\
						\quad & \forall \Zn \subset \inode, \sum_{e: e \in \In{v}, v \in \Zn}\Flow(e) \leq \lambda(z) \\
						\quad & \sum_{e: e \in \Out{\src}}\Flow(e) = \sum_{e: e \in \In{\term}}\Flow(e) = \Cp
\end{alignat*}

\section{Main Results}

\begin{figure*}
    \begin{tabular}{ | l | c | c | c | c |}
    \hline
      & A. Eavesdropper & B. Localized  jammer & C. Localized eavesdropper/jammer & D. Omniscient jammer \\ \hline \hline
   1.1 Na\"ive coding& $\Cp - \dg_{in}(\Zn)$ \cite{cai02secure} & $\Cp-\dg_{out}(\Zn)$ \citep{jaggilangberg05isit} & $\Cp-\dg_{in}(\Zn)-\dg_{out}(\Zn)$ \cite{yao10netcod} & $\Cp -\dg_{out}(\Zn)$ \citep{jaggilangberg05isit} \\
    & & if $\dg_{out}(\Zn)<{\Cp}/{2}$ & if $\dg_{in}(\Zn)+\dg_{out}(\Zn)<{\Cp}$ & if $\dg_{out}(\Zn)<{\Cp}/2$ \\ \hline
 1.2 Toy example & 2 & 0  & 0 & 0  \\ \hline \hline
   2.1 Routing  & $=\Cp -\lambda(z)$ & $ \geq \Cp-\lambda(z)$ & $ \geq \Cp-2\lambda(z)$ & $\geq \Cp - \lambda(z)$ \\
    & & if  $\lambda(z) < \Cp/2$ & if $\lambda(z) < {\Cp}/{2}$ & if  $\lambda(z) < \Cp/2$ \\ \hline
 2.2 Toy example & $8/3$  & $8/3$ & $4/3$  & $8/$3  \\ \hline \hline
   3.1 Coding & $\geq \Cp -\lambda(z)$ \cite{cui10-1} & $\geq \Cp -\lambda(z)$ & $\geq \Cp -2 \lambda(z)$  &  $\geq \Cp -\lambda(z)$ \\   \hline
 3.2 Toy example & $3$  & $3$  & open  & $3$  \\ \hline
 	
    \end{tabular}

\caption{Here, $\lambda(z)$ is the optimal value of the variable $\lambda$ in LP1'. Eavesdropping: In the cockroach network that first describe in \cite{kosut09}, $1$ eavesdropped node can be regarded as $2$ eavesdropping links (since each node has $2$ incoming links). So, the best achievable rate for this example is $2$ by \cite{cai02secure}. In our routing scheme, the rate $\rate = 8/3$ is achievable for the cockroach network example -- see Figure~\ref{fig:illustration}. We Further show that the rate $\rate = 3$ is achievable in the cockroach network example if smart coding is allowed -- see Figure~\ref{fig:coding}. A more general achievable scheme is shown in \cite{cui10-1}. Localized Jamming: The rate for the cockroach network is $0$ if we use the scheme in \cite{jaggilangberg05isit} directly. The rate $\rate = 8/3$ is achievable for the cockroach network -- see the proof of Theorem~\ref{thm:ach:locjam} for the encoding process. The rate $\rate = 3$ is achievable in the example if non-linear coding is allowed -- see Figure~\ref{fig:coding:jam}. (Here the casual omniscient jamming has the same results as the localized jamming -- see \cite{jaggilangberg05isit}). Localized Eavesdropping and Jamming: The rate for the cockroach network is $0$ if we use the scheme in \cite{yao10netcod} directly. The rate $\rate = 4/3$ is achievable if routing in for the cockroach network -- see Remark~\ref{rk:thm:eavesjam} for the encoding process. The coding rate is not known for this case.}
\end{figure*}

We show that the adversarial nodes problem can be solved by routing scheme. The routing is provided by LP1'. We use the same encoding process as \citep{jaggilangberg05isit} in the localized jamming/localized eavesdropping and jamming/omniscient jamming cases. For the localized eavesdropping, we use Vandermonde matrix as the encoding matrix.
\howard{
\begin{theorem}
\label{thm:ach:eaves}
$\rate = \Cp - \lambda(z)$, where $\lambda(z)$ is obtained by an optimal solution from LP1', is achievable for localized eavesdropping.
\end{theorem}
We show that the achievable scheme for localized eavesdropping is optimal.
\begin{theorem}
\label{thm:conv:eaves}
The achievable scheme for localized eavesdropping is optimal among routing schemes.
\end{theorem}
Furthermore, we discovered the graphical properties of the network. The converse for localized eavesdropping against $1$ eavesdropped node can be shown by careful combine the information-theoretic inequalities from its graphical properties.
\begin{theorem}
\label{thm:ach:locjam}
$\rate = \Cp - \lambda(z)$, where $\lambda(z)$ is the variable of LP1', is achievable for localized jamming.
\end{theorem}
\begin{theorem}
\label{thm:ach:eavesjam}
$\rate = \Cp - 2 \lambda(z)$, where $\lambda(z)$ is the variable of LP1', is achievable for localized eavesdropping and jamming.
\end{theorem}
\begin{theorem}
\label{thm:ach:omnijam}
$\rate = \Cp - \lambda(z)$, where $\lambda(z)$ is the variable of LP1', is achievable for omniscient jamming.
\end{theorem}
}

\section{Proofs}
\subsection{Localized Eavesdropping}

\noindent {\it Proof of Theorem~\ref{thm:ach:eaves}}:
By LP1', each path \howard{$p$ is assigned a} flow $\Flow(p)$. It is clear that $\Flow(p)$ is rational for any $p \in \Pset$ since all the coefficients in LP1' are rational. Let $\tau$ be the minimum positive integer such that $\tau \Flow(p) \in \mathbf{Z}^+$. One may consider the scaling factor is scaling the capacity of each link up to $\tau$. Or, one could also consider $\tau$ as the time in a generation. That is, there are $\Cp$ packets transmitted at time $i$ for $i \in \{ 1, 2, \ldots, \tau \}$ and there are $\oblk = \tau \Cp$ packets transmitted to terminal $\term$ in each section. Now, let us consider the following scheme with rate $\rate = \Cp - \lambda$.

{\bf Source:} Let $\vm = \left( m_1, \ldots, m_{\tau \rate} \right)$ be the message transmitted, and $\vk = \left( k_1, \ldots, k_{\tau \lambda} \right)$ be the {\it keys}. The keys are uniformly random over $\Fq$ which is not known to the eavesdropper. So, the messages and the keys are ``embedded" and transmitted over the network and the eavesdropper thus is confused by the random keys. Let $\vanmtrx$ a Vandermonde matrix with size $\oblk \times \oblk$, be the source encoder matrix. Let $\x = (\vm \mbox{ } \vk)^T$ and the information to be transmitted from $\src$ is $\vanmtrx \x$. So, each packet corresponds to an entry of $\vanmtrx \x$.

{\bf Intermediate Nodes:} The packets are transmitted via the routes given by LP1'.

{\bf Terminal:} At terminal $\term$, the terminal $\term$ simply multiplies $\vanmtrx^{-1}$ with the received information $\vanmtrx \x$. Hence, $\x$ is recovered.

For any $\Zn \subset \inode$, the total amount of flows passing through $\Zn$ is at most $\tau \lambda$. There are also $\tau \lambda$ uniform random numbers that are not known by the eavesdropper. Thus, the eavesdropper is not able to get any information of the original message no matter which set of $\Zn$ nodes he observes. Therefore, the rate $\rate = \Cp - \lambda$ is achievable by the above scheme. \endIEEEproof

\noindent {\it Proof of Theorem~\ref{thm:conv:eaves}}:

{\bf Step 1:} We first show that there is a routing scheme {\it without} replicating that performs at least as well as any routing scheme {\it with replicating}.\footnote{We defined replicating routing schemes as those in which an internal node transmits the same incoming packet at least twice on outgoing edges.} Suppose there is a node $v \in \inode$ that performs replicating. Consider the routing scheme obtained by removing all but one of the replicated packets from the network (keeping only one of those reaching the terminal, if there is one such, else removing all of the packets). Under this new routing scheme, the information received by the terminal still enables it to reconstruct as well as under the previous scheme. In addition, removing packets from the network can only improve the secrecy requirement. Sequentially removing all replicated packets thus results in a routing-without-replicating scheme with performance at least as good as the original scheme.

Next, we give a more nuanced argument to show that in fact, for an optimal routing scheme, even the packets leaving the source must be essentially (statistically) independent. Let $p_1, p_2, \ldots, p_k$ be all the paths from the source $\src$ to the terminal $\term$. Let $\Pvar(j)$ be the random variable transmitted on the path $p_j$. So, for the paths $p_j \ni e$, \howard{we have $\entropy{\Pvar(j), j:p_j \ni e} \leq 1$. We assume the secrecy $\mutual{\Msg}{\Pvar(\Zn)} \leq \epsilon_2$ and the probability of decoding error $\Pr_e = \Pr_{\Act, \Code}(\widehat{\Msg} \neq \Msg) \leq \epsilon_1$.} By the Slepian-Wolf Theorem~\cite{slepianwolf}, we can construct a new random variable $\widehat{\Pvar}(j)$ for each path $p_j$ from the source $\src$ to the terminal $\term$ with certain properties. Firstly, the set $\{\widehat{\Pvar}(j)\}$ still carries essentially all the information that the set or original random variables $\{{\Pvar}(j)\}$ carried, and hence the terminal can still decode $\Msg$. Second, each $\widehat{\Pvar}(j)$ is a function {\it only} of ${\Pvar}(j)$, and hence the new routing scheme divulges no more information to the eavesdropper than the original scheme (due to the data-processing inequality). Third, the individual entropies of each new random variable is no more than the entropy of the original random variable, hence the edge-capacity constraints are not violated by the new routing scheme. Finally, the joint entropy of the new random variables is essentially the same as the sums of their individual entropies. Specifically, for any $\epsilon' >0$, there is a sufficiently large $m$ ({\it number of generations}), such that
$\sum_{j=1}^k \entropy{\widehat{\Pvar}(j)} = \entropy{\Pvar(1)^m, \ldots, \Pvar(k)^m} + m\epsilon'$.
For each $j$, the specific choice of $\widehat{\Pvar}(j)$ that satisfies these constraints simultaneously corresponds to the output of the $j$-th Slepian-Wolf source encoder operating at any rate-point on the sum-rate constraint of Slepian-Wolf rate-region.

\howard{
{\bf Step 2:}
We now use the properties of the new routing scheme derived in Step 1 to argue that in fact the rate specified by the solution of LP1 is also an outer bound on the achievable rate for routing-only schemes.
	\begin{eqnarray}
		m \rate &=& \entropy{\Msg^m} \\	
		& \leq & \entropy{\Msg^m, \widehat{\Pvar}(\Zn)^m} \\	
		& \leq & \condentropy{\Msg^m}{\widehat{\Pvar}(\Zn)^m} + \mutual{\Msg^m}{\widehat{\Pvar}(\Zn)^m} \\
		& \leq & \condentropy{\Msg^m}{\widehat{\Pvar}(\Zn)^m} - \condentropy{\Msg^m}{\widehat{\Pvar}(\Zn)^m, \overline{\widehat{\Pvar}(\Zn)}^m} \nonumber \\
		& & + \condentropy{\Msg^m}{\widehat{\Pvar}(\Zn)^m, \overline{\widehat{\Pvar}(\Zn)}^m} + m \epsilon_1 \\
		& \leq & \condmutual{\Msg}{\overline{\widehat{\Pvar}(\Zn)}^m}{\widehat{\Pvar}(\Zn)^m} \nonumber \\
		& & + 1 + \epsilon m \rate + m \Delta \label{eq:fano} \\
		& \leq & \entropy{\overline{\widehat{\Pvar}(\Zn)}^m} + 1 + \epsilon m \rate + m \Delta \\
		&=& m \Cp - m \entropy{\widehat{\Pvar}(\Zn)} \nonumber \\
		& & - m \epsilon' + 1 + \epsilon m \rate + m \Delta
	\end{eqnarray}
	\noindent where $\overline{\widehat{\Pvar}(\Zn)}$ denotes the random variables $\{ \widehat{\Pvar}(1), \ldots, \widehat{\Pvar}(k) \} \setminus \widehat{\Pvar}(\Zn)$.
Inequality~(\ref{eq:fano}) holds by Fano's inequality, and the last equality holds due to the ``near-independence'' of $\widehat{\Pvar}(j)$, as argued in Step 1 (the remaining steps follow from standard information identities and inequalities).
Hence $\rate \leq \frac{1}{1 - \epsilon}\left[\Cp - \entropy{\widehat{\Pvar}(\Zn)} - \epsilon' + \frac{1}{m} + \Delta\right]$. But this entropy inequality must hold for each set $\Zn$. But these, along with the entropy inequalities constraining the rate on each edge to be at most $1$, match the corresponding achievable rate given by LP1. \endIEEEproof
}
{
\subsection{Alternate outer bound for $z=1$}

We now present an alternative proof technique for the outer bound on the rate in the scenario when the network has just a single node-based eavesdropper. This technique provides an interesting graphical characterization of optimal routing-based schemes. Unfortunately this technique, as presented, does not extend to the case when $z > 1$, nor when coding is allowed inside the network. Nonetheless, we are hopeful that one or both of these limitations may be overcome if our techniques are combined with a more careful analysis of structured information inequalities, such as those presented in Madiman-Tetali~\cite{madiman10}.

\begin{definition}[Node-cut]
A set of nodes $\ncut \subset \vertex$ is called a node-cut if after removing the nodes in $\ncut$ there does not exist any path from $s$ to $t$ in the network.
\end{definition}
Of particular interest are minimal node-cuts.
\begin{definition}[Minimal node-cut]
A set of nodes $\ncut \subset \vertex$ is called a minimal node-cut if $\ncut$ is a node-cut and no proper subset of $\ncut$ is a node-cut. The set of all minimal node-cuts is denoted by $\widehat{\ncut}$.
\end{definition}

We first show the existence of a minimal node-cut satisfying certain properties. Specifically, we show that each node in this node-cut must be either {\it capacity constrained} (the flow passing through the node is constrained by the capacities of incoming outgoing edges) or {\it secrecy constrained} (the flow passing through the node is constrained by the requirement that there be no information-leakage if that node is eavesdropped on).
Such a node-cut, combined with carefully chosen information inequalities, is used to obtain an information theoretic upper bound on the capacity of the network. The scheme in LP1 achieves this upper bound.

For a minimal node-cut $\widehat{\ncut}$ we define the following sets.
\begin{align}
\mathcal E(\widehat{\ncut})&\triangleq \left\{e=(u,v)\in \mathcal E: u,v \in \widehat{\ncut}\right\} \nonumber \\
\widehat{\ncut}_{\lambda}&\triangleq \left\{v \in \widehat{\ncut}: \sum_{p:p\ni v}f(p)={\lambda}\right\}\nonumber \\
\widehat{\ncut}_{\Cp}&\triangleq \lbrace v \in \widehat{\ncut}: v \not \in \widehat{\ncut}_{\lambda}, \sum_{p:p \ni v}f(p)=\nonumber \\
& \min\left\{|\{e \in \In{v}\setminus\mathcal E(\widehat{\ncut})\}|, |\{e \in \Out{v}\setminus\mathcal E(\widehat{\ncut})\}\right\} \nonumber \rbrace
\end{align}

\begin{lemma}\label{existncut}
For a given single-source single sink network $\graph$, there exists a minimal node-cut $\widehat{\ncut}$ such that
\begin{align}\label{eq:node-cut prop}
\widehat{\ncut} = \widehat{\ncut}_\lambda \cup \widehat{\ncut}_\Cp.
\end{align}
\end{lemma}

\begin{proof}
We prove the lemma by contradiction. Assume there does not exists a minimal node-cut in the given network with property \eqref{eq:node-cut prop}. Consider the minimal node-cut $\widehat{\ncut}=\{v: \exists e=(s,v) \in \mathcal E\}$ and define the set $\mathcal{U}(\widehat{\ncut})=\{u \in \widehat{\ncut}: u \not \in \widehat{\ncut}_\lambda \cup \widehat{\ncut}_\Cp\}$.
Then there exists a node $u \in N$ such that $u \not \in \widehat{\ncut}_\lambda \cup \widehat{\ncut}_\Cp$.

Now consider the set $tail(\Out{\mathcal U})\cup \widehat{\ncut} \setminus \mathcal U$ and choose any minimal node-cut $\widehat{\ncut}' \subseteq tail(\Out{\mathcal U})\cup \widehat{\ncut} \setminus \mathcal U$. By assumption, there must exist some non-empty set $\mathcal U(\widehat{\ncut}')=\{u \in \widehat{\ncut}': u \not \in \widehat{\ncut}_\lambda' \cup \widehat{\ncut}_\Cp'\}$. Repeat the process of finding new minimal node-cut until we find a node $w$ in a new node-cut $\widehat{\ncut}''$ such that there exist edge $(w,t)$ and $w \not \in \widehat{\ncut}_\lambda'' \cup \widehat{\ncut}_\Cp''$.

Note that the above process reveals existence of a path $p: s,v,\ldots,w,t$ such that $f(p)<1$, which implies $\sum_{p \in \Pset} f(p) < \Cp$, which is a contradiction.
\end{proof}

\noindent {\it Alternative proof of the outer bound.} By Lemma~\ref{existncut}, let $\widehat{\ncut}$ be the node-cut we consider in the network $\graph$. Note that for any node $v \in \widehat{\ncut}$, we have the following sequence of inequalities:
\begin{eqnarray}
	\rate &=& \entropy{\Msg} \\
		& \leq & \condentropy{\Msg}{\X(\In{\widehat{\ncut}} \setminus \edge(\widehat{\ncut}))} \nonumber \\
		& & + \mutual{\Msg}{\X(\In{\widehat{\ncut}} \setminus \edge(\widehat{\ncut})} \\
		&=& \mutual{\Msg}{\X(\In{\widehat{\ncut}} \setminus \edge(\widehat{\ncut})} \label{eq:decod}\\
		&=& \condmutual{\Msg}{\X(\In{\widehat{\ncut} \setminus \{ v \}} \setminus \edge(\widehat{\ncut})}{\X(\In{v} \setminus \edge(\widehat{\ncut}))} \nonumber \\
		& & + \mutual{\Msg}{\X(\In{v} \setminus \edge(\widehat{\ncut}))}  \\
		& \leq & \condentropy{\X(\In{\widehat{\ncut} \setminus \{ v \}} \setminus \edge(\widehat{\ncut})}{\X(\In{v} \setminus \edge(\widehat{\ncut}))} \nonumber \\
		& & \label{eq:secr} \\
		&=& \entropy{\X(\In{\widehat{\ncut}} \setminus \edge(\widehat{\ncut}))} \nonumber \\
		& & - \entropy{\X(\In{v} \setminus \edge(\widehat{\ncut}))}
\end{eqnarray}
\noindent Here \eqref{eq:decod} follows from the requirement that the message $\Msg$ be decodable from the network transmissions, \eqref{eq:secr} from the requirement that there be no information leakage, and the remaining are standard information identities and inequalities.

\noindent Summing up the above inequalities for every $v \in \widehat{\ncut}_\lambda$, we have
\begin{eqnarray}
	|\widehat{\ncut}_\lambda | \rate & \leq & |\widehat{\ncut}_\lambda | \entropy{\X(\In{\widehat{\ncut}} \setminus \edge(\widehat{\ncut}))} \nonumber \\
	& & - \sum_{v \in \widehat{\ncut}_\lambda} \entropy{\X(\In{v} \setminus \edge(\widehat{\ncut}))}
\end{eqnarray}
Note that
\begin{eqnarray}
	& & \sum_{v \in \widehat{\ncut}_\lambda} \entropy{\X(\In{v} \setminus \edge(\widehat{\ncut}))} \nonumber \\
	& & + \sum_{v \in \widehat{\ncut}_\Cp} \entropy{\X(\In{v} \setminus \edge(\widehat{\ncut}))}  \nonumber \\
	&\geq & \entropy{\X(\In{\widehat{\ncut}} \setminus \edge(\widehat{\ncut}))} \label{eq:conv-main}
\end{eqnarray}
So, we have
\begin{eqnarray}
	|\widehat{\ncut}_\lambda | \rate & \leq & |\widehat{\ncut}_\lambda | \entropy{\X(\In{\widehat{\ncut}} \setminus \edge(\widehat{\ncut}))} \nonumber \\
	& & - \sum_{v \in \widehat{\ncut}_\lambda} \entropy{\X(\In{v} \setminus \edge(\widehat{\ncut}))} \\
	& \leq & \left(|\widehat{\ncut}_\lambda | - 1\right) \entropy{\X(\In{\widehat{\ncut}} \setminus \edge(\widehat{\ncut}))} \nonumber \\
	& & + \sum_{v \in \widehat{\ncut}_\Cp} \entropy{\X(\In{v} \setminus \edge(\widehat{\ncut}))} \\
	&=& \left(|\widehat{\ncut}_\lambda | - 1\right) \Cp \nonumber \\
	& & + \sum_{v \in \widehat{\ncut}_\Cp} \left| \X(\In{v} \setminus \edge(\widehat{\ncut})) \right|
\end{eqnarray}
Therefore,
\begin{equation}
\label{eq:conv}
	\rate \leq \left( 1 - \frac{1}{|\widehat{\ncut}_\lambda |} \right) \Cp + \frac{1}{|\widehat{\ncut}_\lambda |} \sum_{v \in \widehat{\ncut}_\Cp} \left| \X(\In{v} \setminus \edge(\widehat{\ncut})) \right|
\end{equation}

Note that \eqref{eq:conv-main} is equivalent to $|\widehat{\ncut}_\lambda | \lambda + \sum_{v \in \widehat{\ncut}_\Cp} \left| \X(\In{v} \setminus \edge(\widehat{\ncut})) \right| \geq  \Cp$. Hence, $\rate \leq \Cp - \lambda$ can be verified by putting \eqref{eq:conv-main} into \eqref{eq:conv}. \endIEEEproof
}
}

\subsection{Localized Jamming}

\noindent {\it Proof of Theorem~\ref{thm:ach:locjam}}:

We use the same achievable scheme as \cite{jaggilangberg05isit} for the localized jamming scenario. Roughly speaking, each packet contains $3$ parts in this achievable scheme. That is, information about the message, a seed of hash function, and the value of the hash.

{\bf Source:} First, the source $\src$ fixes a number $\rate' = \lfloor \left( 1 - \frac{\oblk + 1}{\blk} \right)(\oblk - \tau \lambda) \rfloor$. Let the source encoder matrix be a Vandermonde matrix $\vanmtrx$ with size $(\blk - \oblk - 1) \oblk \times \blk \rate'$. Let $\x$ be the vector of the original message. The vector $\x$ is of dimension $\blk \rate'$. There are $\tau$ timeslots in one section, where $\tau$ is the minimum positive integer such that $\tau \Flow(p) \in \mathbb{Z}^+$. In each section, the source $\src$ is transmitting $\oblk = \tau \Cp$ packets to the terminal $\term$. More precisely, $\Cp$ packets are transmitted to the terminal $\term$ at time $i \in \{ 1, 2, \ldots, \tau \}$. Let us denote $p_1, p_2, \ldots, p_\oblk$ be the packets that the source $\src$ transmits to the terminal $\term$. We also denote the corresponding encoding matrix for the packet $p_j$ to be $\vanmtrx(p_j)$. The size of the matrix $\vanmtrx(p_j)$ equals $(\blk - \oblk - 1) \times \blk \rate'$. Also note that the concatenation of $\vanmtrx(p_j)$ for all $j \in \{1,2, \ldots, \oblk\}$ is exactly the encoding matrix $\vanmtrx$. Let $\rho$ be a random number that uniformly chosen in $\Fq$. Define $\T(p_j)$, $\U$ and $\D$ as follows,
\begin{equation}
	\T(p_j) = \vanmtrx(p_j) \x
\end{equation}
\begin{equation}
	\U = [1, \rho, \rho^2, \ldots, \rho^{\blk - \oblk - 1}]
\end{equation}
\begin{equation}
	\D = \U [\T(p_1) \mbox{ } \T(p_2) \mbox{ } \ldots \mbox{ } \T(p_\oblk)]
\end{equation}
Where $\T(p_j)$ is the vector that contains the information about the message, $\U$ is the hash function with respect to $\rho$, and $\D$ is the value of the hash. So, $\T(p_j)$ is the column vector of size $\blk - \oblk - 1$, $\D$ is a row vector of size $\oblk$. Let the whole packet $p_j$ to be $[\T(p_j) \mbox{ } \D \mbox{ } \rho]$. Clearly, the packet size is $\blk$, and the last symbol of the packets $\rho$ is the seed of the hash.

{\bf Intermediate Nodes:} The packets $p_j$ are transmitted from the source $\src$ to the terminal $\term$ follows the corresponding paths. That is, intermediate nodes $v \in \inode$ preform routing (not replicating) that is given by LP1'.

{\bf Terminal:} At terminal $\term$, the decoding procedure is the following. Let the terminal receives $[\widehat{\T}(p_j) \mbox{ } \widehat{\D} \mbox{ } \widehat{\rho}]$ for $j \in \{ 1, 2, \ldots, \oblk \}$. The terminal $\term$ first determines $[\D' \mbox{ } \rho']$ by choosing the majority of received packets in a section. Since $\rho'$ is now fixed, we have $\U'$ is also fixed. Then, the terminal $\term$ checks whether $\U' \widehat{\T}(p_j)$ equals the $j$-th symbol of $D'$. Denote this set of packets to be $\Pset_D$.

Next, the terminal $\term$ concatenates the matrices $\vanmtrx(p_j)$ for $p_j \in \Pset_D$ into a matrix, denoted by $\vanmtrx_D$. Note that
\begin{equation}
	\begin{pmatrix}
		\widehat{\T}(p_1') \mbox{ } \widehat{\T}(p_2') \mbox{ } \cdots \mbox{ } \widehat{\T}(p_{|\Pset_D|}')
	\end{pmatrix}^T = \vanmtrx_D \x
\end{equation}
where $p'$ corresponds to the packets in $\Pset_D$. So, $\x$ can be founded by inverting the matrix $\vanmtrx_D$. The probability of decoding error equals $1 - \blk \oblk 2^{-m}$ is shown in \cite{jaggilangberg05isit}, where $m$ is the field size parameter. \endIEEEproof


\begin{remark}
If nodes in $\vertex$ are not allowed to replicate incoming packets to outgoing links, the achievable rate is indeed optimal -- see~\cite{jaggilangberg05isit}. If nodes in $\vertex$ are allowed to replicate incoming packets to outgoing links, the achievable rate can be improved -- see Figure~\ref{fig:counter}. 
\end{remark}

\begin{figure}
  \centering
  \includegraphics[width=3in]{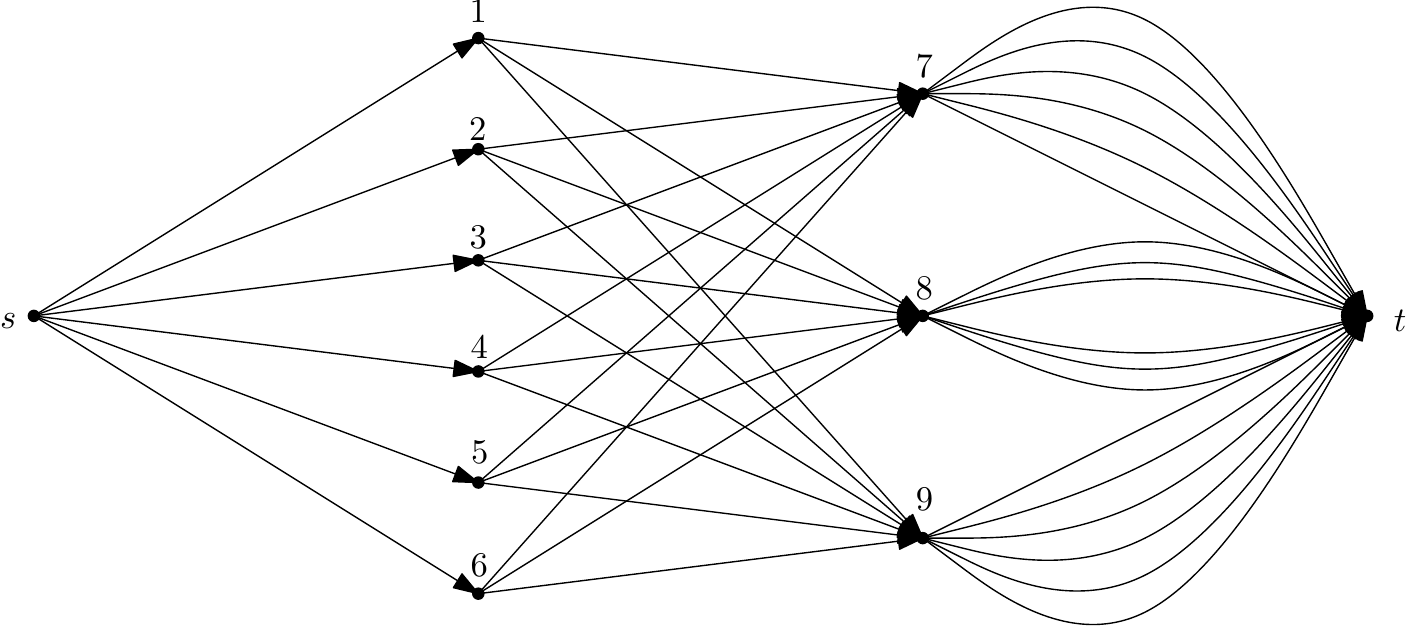}
  \caption{An example that shows that higher rates may be achievable in the {\bf Causal Omniscient Jamming} model than are achievable {\bf Localized Eavesdropping} model. This is contrary to the behavior one sees in the link-adversary case -- see, for instance,~\citep{sid-book}, and is thus somewhat surprising.
 The example requires nodes inside the network to perform replication -- as shown in the characterization of the capacity of the {\bf Localized Eavesdropping} model, in-network replication does not help improve the rate. Suppose one of the nodes in the network is a causal omniscient jammer.
Further suppose that the source $\src$ uses the same encoding as Theorem~\ref{thm:ach:locjam}. The nodes in the first layer replicate the packets and send out identical copies on each outgoing link. So, each node in the second layer receives the same set of packets, and each then forwards their outgoing packets. If one of the nodes from the second layer is jammed, terminal $\term$ can decode correctly without any rate loss by majority decoding. So, the optimal adversarial strategy is to jam a node in the first layer. Therefore, a rate $\rate = 5$ is achievable. However, solving LP1 shows us that the optimal rate achievable in the {\bf Localized Eavesdropping} model is $4$.}
  \label{fig:counter}
\end{figure}

\begin{remark}
\label{rk:thm:eavesjam}
For the proof of Theorem~\ref{thm:ach:eavesjam}, the only difference between the proof of Theorem~\ref{thm:ach:locjam} is that $\x = (\vm \mbox{ } \vk)$ where $\vm$ is the message with size $\blk \rate''$ in which $\rate'' = \lfloor \left( 1 - \frac{\oblk + 1}{\blk} \right)(\oblk - 2 \tau \lambda) \rfloor$ and $\vk$ is the key with size $\blk \lfloor \left( 1 - \frac{\oblk + 1}{\blk} \right)\tau \lambda \rfloor $.
\end{remark}
\begin{remark}
The proof of Theorem~\ref{thm:ach:omnijam} is the same as the proof of Theorem~\ref{thm:ach:locjam}.
\end{remark}

\subsection{Encoding Complexity versus Rate-optimal Loss}

Note that the size of encoding matrix is determined by $\oblk = \tau \Cp$. Since $\tau$ is the parameter determined by LP1', the encoding complexity is large when $\tau$ is also large. In this section, we will give the rate loss when we fix $\tau$.

\begin{lemma}
For $\tau'$ fixed, denote the corresponding rate to be $\rate'$. We have $\rate - \rate' < \frac{| \edge |}{\tau'}$.
\end{lemma}
\begin{proof}
Solving the {bf Linear Program 2} of network $\graph$ gives us the flow value $\Flow(e)$ on each link $e \in \edge$. Reduce the network $\graph$ by setting each link to be capacity $\Flow(e)$ and multiply $\tau'$ to each link of the network $\graph$. So, each link has capacity equals to $\tau' \Flow(e)$, denote this scaled network to be $\graph'$. Denote the network $\graph''$ to be the quantization on each link $e$ to be integer value, {\it i.e.}, taking $\lfloor \tau' \Flow(e) \rfloor$. So, the capacity on each link $e$ is reduced by a value at most $1$. Therefore, the capacity of the network $\graph''$ is reduced at most $| \edge |$ from the network $\graph'$. Therefore, the capacity of the network $\graph$ by fixing $\tau'$ reduced is at most $\frac{| \edge |}{\tau'}$. Hence, $\rate - \rate' < \frac{| \edge |}{\tau'}$.
\end{proof}

\begin{figure}
  \centering
  \includegraphics[width=3in]{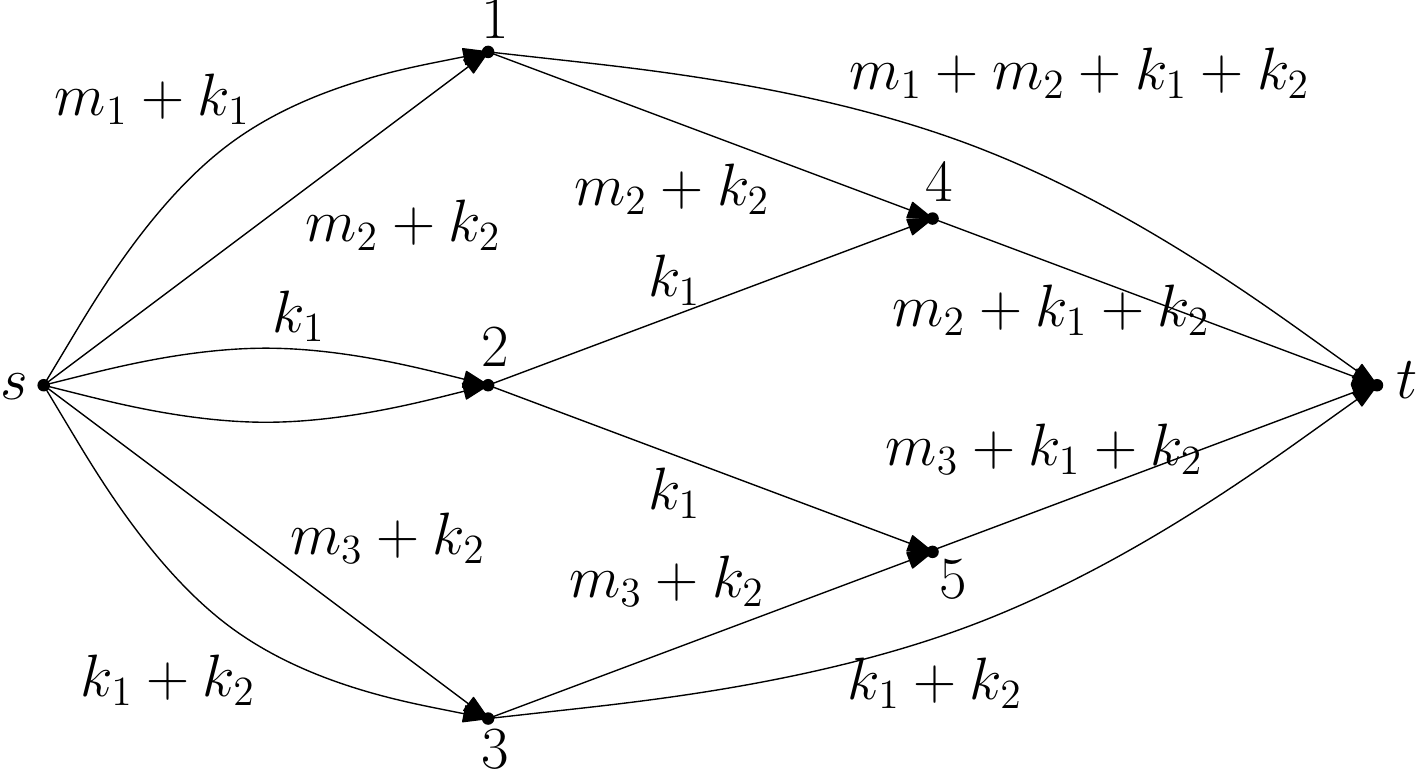}
  \caption{An example, for the cockroach network, of a ``careful coding'' scheme that beats any routing scheme in the {\bf Localized Eavesdropping} model. It can be verified by solving {\bf Linear Program 1} that the routing-only rate equals $8/3$. However, it can be verified that in the scenario that $z = 1$, {\it i.e.}, at most one node is eavesdropped on, the scheme outlined in this figure ensures that a rate $\rate$ of $3$ is perfectly securely achievable.}
  \label{fig:coding}
\end{figure}
\begin{figure}
  \centering
  \includegraphics[width=3in]{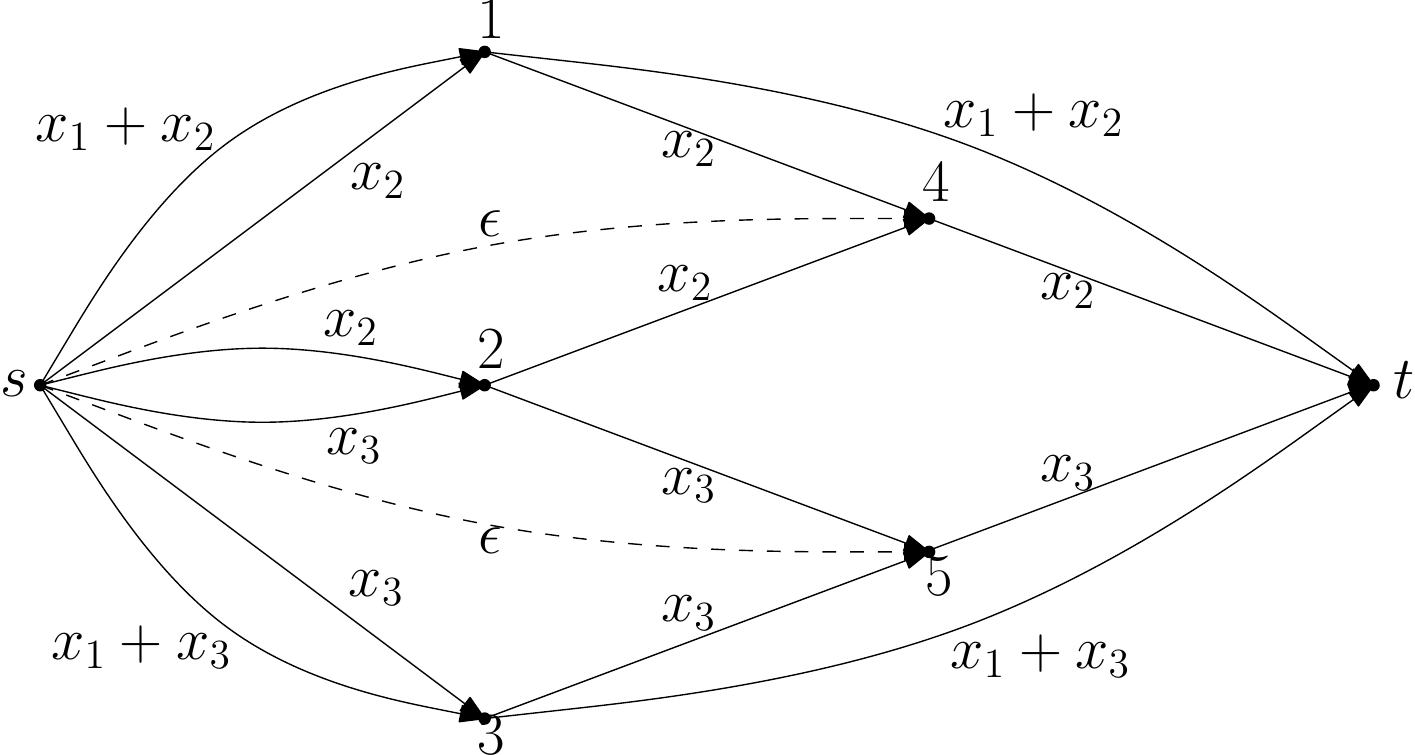}
  \caption{An example demonstrating that allowing coding operations inside the network can in general leads to capacities that are greater than are possible by the routing-only schemes we present in this work. Specifically, suppose the cockroach network was augmented by a link with ``small'' capacity to each of nodes $4$ and $5$. In this case, as in~\cite{jaggilangberg05isit} the (honest) source can send check-sums of the packets that should have reached each of nodes $4$ and $5$ via other routes in previous generations -- if these do not match the packets actually received by those nodes, they can discard these packets and forward the ``useful'' packets, leading to an achievable rate of $3$. In contrast, our routing schemes can achieve a rate of at most $8/3 + \epsilon$.}
  \label{fig:coding:jam}
\end{figure}

\section{Beyond routing}
\label{sec:beyroute}
In this Section we demonstrate that carefully chosen network codes can indeed outperform many of the routing-only schemes presented as some of the main results of this work. However, the complexity of designing and implementing these schemes is in general much higher than that of the routing schemes we focus on. Also, a
complete characterization the optimal throughout of such schemes is still open, and is thus, in the main, left open in this work.
\begin{figure}
  \centering
  \includegraphics[width=3in]{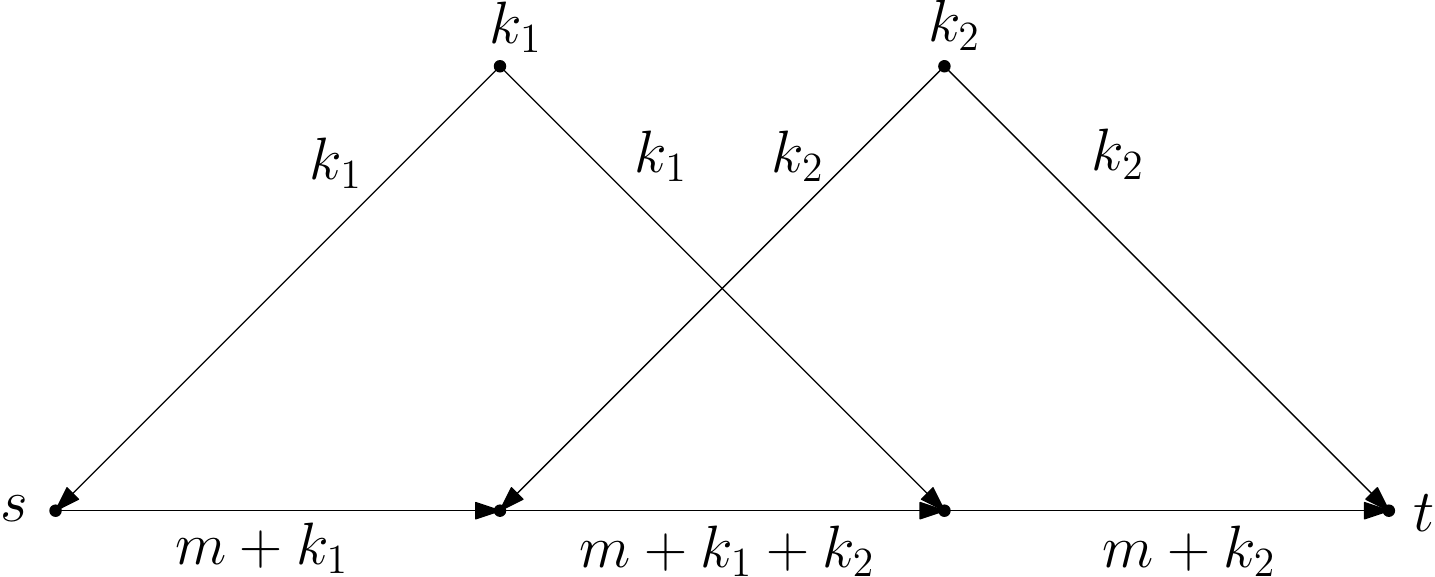}
  \caption{An example demonstrating that in some scenarios, if nodes inside the network are allowed to inject randomness, higher rates can be achieved than if this is not allowed (this observation was previously made in~\cite{cui10-1} -- we repeat it here with a simpler example).}
  \label{fig:int_rand}
\end{figure}

\bibliographystyle{IEEEtran}
\bibliography{nwc}

\end{document}